\documentclass[preprint,12pt]{elsarticle}
\usepackage[utf8]{inputenc}
\usepackage{graphicx}
\usepackage{amssymb}
\usepackage{multicol}
\usepackage{color}
\usepackage{amsthm}
\usepackage{amsfonts,amssymb,amscd,amsmath,euscript,enumerate,verbatim,calc}
\newtheorem{construction}{Construction}
\newtheorem{definition}{Definition}
\newtheorem{theorem}{Theorem}
\newtheorem{proposition}{Proposition}
\newtheorem{corollary}{Corollary}
\newtheorem{example}{Example}
\usepackage{url} 

\newenvironment{sproof}{%
	\proof}{\endproof}
\usepackage[normalem]{ulem}

\newcommand{\rmv}[1]{}                                
\newcommand{\todo}[1]{{#1}} 			    

 \journal{journal (final version)}
 
\begin{document}
 \begin{frontmatter}
\title{Nested Cover-Free Families for Unbounded Fault-Tolerant Aggregate Signatures}

\date{December 2020} 
 \author{Thais Bardini Idalino}
 \ead{tbardini@uottawa.ca}
 \author{Lucia Moura}
  \ead{lmoura@uottawa.ca}
 \address{University of Ottawa}

\begin{abstract}
Aggregate signatures are used to create one short proof of authenticity and integrity from a set of digital signatures. However, one invalid signature in the set invalidates the entire aggregate, giving no information on which signatures are valid. Hartung et al. (2016) propose a fault-tolerant aggregate signature scheme based on combinatorial group testing. 
Given a bound $d$ on the number of invalid signatures among $n$ signatures to be aggregated, this scheme uses $d$-cover-free families to determine which signatures are invalid.
These combinatorial structures guarantee a moderate increase on the size of the aggregate signature that can reach the best possible compression ratio 
of  $O(\frac{n}{\log n})$, for fixed $d$, coming from an information theoretical bound.
The case where the total number of signatures grows dynamically (unbounded scheme) was not satisfactorily solved in their original paper, since explicit constructions had constant compression ratios. In the present paper, we propose efficient solutions for the unbounded scheme, relying on sequences of $d$-cover-free families that we call {\em nested families}. Some of our constructions yield high compression ratio close to \rmv{the information theoretical bound}\todo{the best known upper bound}. We also propose the use of $(d,\lambda)$-cover-free families to support the
loss of up to $\lambda-1$ parts of the aggregate.
\end{abstract}

\end{frontmatter}

\section{Introduction}\label{sec:intro}
In cryptography,  aggregate signature schemes allow us to combine a set of digital signatures into a single one, which can be used as proof of integrity and authenticity of a possibly large set of data. This solution is specially useful for applications that manage a large quantity of data and digital signatures, since it can save on communication and storage, as well as improve the signature verification process. A few examples of such applications are outsourced databases \cite{Mykletun}, sensor networks \cite{Li2010}, secure logging \cite{Ma}, certificate chains \cite{Boneh}, vehicular communication \cite{Wasef}, among others. 

The verification of an aggregate signature outputs a positive result only if the entire set of signatures is valid. If we have at least one faulty signature in the set, the proof of integrity and authenticity of all the data involved is invalidated. This happens because when a set of signatures is aggregated into one, this operation does not preserve enough information in order to identify the exact set of invalid signatures. In order to address this shortcoming, Hartung et al.~\cite{fault} propose a fault-tolerant scheme using $d$-cover-free families ($d$-CFF). This scheme generates a more expressive aggregate signature, that can tolerate up to $d$ invalid signatures and identify all the valid ones. 

Combinatorial group testing \cite{gtLivro} deals precisely with this type of problem of determining a set of defective items via testing groups of items.
A $d$-CFF  is a collection of subsets such that no subset is contained in the union of any $d$ other sets in the family.  A 2-CFF is shown in Fig.~\ref{2cffeg}, where each column represents the incidence vector of a set; in group testing, the rows of the matrix represents the groups and its columns, the items. By inspection, we can see that if up 2 items are defective, this set of defectives can be determined from the outcomes of the group tests. For example, if only items $1$ and $2$ are defective, tests $1-6$ fail and tests $7-9$ pass, and the defectives can be determined from this outcome.

\begin{figure}[h!]\label{2cffeg}
\[\setcounter{MaxMatrixCols}{12}
\mathcal{M} = \begin{pmatrix}\label{2cff}
1& 0& 0& 1& 0& 0&  1&0 &0& 1&0& 0\\
1 & 0 & 0 & 0 &1  & 0 & 0 &  1 & 0 &  0 & 1&  0 \\
1 & 0 & 0 & 0 & 0 &  1 & 0 &0 & 1 &0 & 0 & 1 \\
0 &  1& 0 & 1 & 0 & 0 & 0 &0 & 1 &0 & 1 & 0 \\
0 & 1 & 0 & 0 & 1 & 0 & 1 & 0 & 0 & 0 &  0 & 1\\
0 & 1 & 0 & 0 & 0 & 1 & 0 &1& 0 & 1 & 0 &  0\\
0 & 0 &  1 &  1 & 0 & 0 & 0 &  1 & 0 & 0 & 0 &  1 \\
0 & 0 &  1 & 0 &  1 & 0 & 0 & 0 &  1&  1 & 0 & 0 \\
0 & 0 &  1 & 0 & 0 & 1 & 1 &0 & 0 & 0 & 1 & 0  \\
\end{pmatrix}
\]
\caption{A $2$-cover-free family $2$-CFF($9,12$).}
\end{figure}

What is distinctive in the signature aggregation application is that the groups are formed at a different stage than testing: groups of signatures are combined at the aggregation stage  and their testing is done at the verification stage. Another distinctive characteristic of signature aggregation is that we do not have a bound on the number of signatures that need to be aggregated; for example, secure logging applications can not always predict how many log entries will be saved, and outsourced databases not always have control on the amount of data that will be inserted.
 This requires an unbounded fault-tolerant aggregation scheme relying on matrix growth, which, as pointed out by Hartung et al.~\cite{fault}, requires a special sequence of $d$-CFFs arrays rather than a single one. 
Here, the concept of compression ratio is important.  Consider the total number of signatures $n$ (number of columns) and the size of the aggregate signature $s(n)$ (number of rows)
at each stage. The compression ratio is equal to $\rho(n)$ iff  $\frac{n}{s(n)}$ is  $\Theta(\rho(n))$. \rmv{Information theoretical}\todo{The best known} lower bounds on $s(n)$ ~\cite{furedi} show that $\rho(n)$ 
is $O(\frac{n}{(d^2/\log d) \log n})$.

Monotone families of $d$-CFFs  were defined in~\cite{fault} to accomplish unbounded aggregation. 
 The explicit constructions of monotone families given in~\cite{fault}  yield aggregate signatures with length linear in $n$ and thus \emph{constant} compression ratios;
 the authors  pose the open problem of finding better monotone families in order to achieve a more efficient unbounded scheme. In~\cite{COCOApaper}, the present authors 
 find monotone families with $\rho(n)=n^{1-1/c}$, for every integer $c\geq 2$ and fixed $d\geq 1$. 
 In the present work, we approach unbounded aggregate signatures by defining a more flexible sequence of $d-$CFFs that we call a \emph{nested family}, and which generalizes monotone families. We then propose various constructions of nested families, some with compression ratio close to the information theoretical upper bound.
 Although we focus on the specific problem of fault-tolerant digital signature aggregation, our new constructions can be applied to other  combinatorial group testing problems where the number of items to be tested grows dynamically. An earlier version of the present work appeared in~\cite{nested}.
 
In Section 2, we present the background on fault-tolerant aggregate signatures introduced in~\cite{fault} and on $d$-CFF constructions.
In Section 3, we define nested families and present unbounded signature aggregation algorithms that use these new families. 
In Section 4, we provide explicit constructions of such families that yield unbounded aggregate schemes with better compression ratios than previously known. 
In Section 5, we generalize our $d$-CFF constructions to build ($d;\lambda$)-CFFs, which are defined there, and discuss how they can also support faults on signature transmission.  These extensions in Section 5 are not \rmv{mentioned}\todo{considered} in the early version of this paper~\cite{nested}. Section 6 includes conclusions and future work.

\section{Background on Fault-Tolerant Schemes}\label{sec:background}

In this section, we present the fault-tolerant aggregate signature scheme by Hartung et al.~\cite{fault}. We summarize the concepts introduced by the authors~\cite[Sections 1,3,4]{fault}, including the necessary formalization to contextualize our construction. Then, we present background and recursive constructions of cover-free families.

{\todo{\subsection{Traditional signature aggregation}} 
\noindent
\rmv{\noindent {\em 2.1 Fault-Tolerant Aggregate Signature}\\}
Let $\{\sigma_1, \ldots, \sigma_n\}$ be a set of $n$ signatures and let \todo{the set $C = \{(pk_1, m_1), \ldots,\\(pk_n, m_n)\}$} be their corresponding pairs of public key and message. 
A traditional aggregate signature scheme consists of combining all $n$ signatures together in one aggregate signature $\sigma$, which can be as small as a single digital signature \cite{Boneh}. 
By verifying only $\sigma$ we can ensure the integrity and authenticity of the entire set  $C$.
If all signatures $\{\sigma_1, \ldots, \sigma_n\}$ are correctly formed from $C$, the signature verification outputs 1, but if at least one $\sigma_i$ does not match its corresponding pair ($pk_i, m_i$), the verification outputs 0.
In more details,  Boneh et al.~\cite{Boneh} define an aggregate signature scheme with four algorithms \cite{Boneh,fault}:
	
	\begin{enumerate}
		\item \textbf{KeyGen}($1^\kappa$) takes security parameter $\kappa$ and creates a key pair $(pk, sk)$.
		\item \textbf{Sign}($sk, m$) creates a signature for message $m$ using secret key $sk$.
		\item \textbf{Agg}($C_1, C_2, \sigma_1, \sigma_2$) takes two multisets of public key and message pairs $C_1$ and $C_2$ and their corresponding signatures $\sigma_1$ and $\sigma_2$, and outputs an aggregate signature $\sigma$ that certifies the integrity and authenticity of $C = C_1 \cup C_2$.
		\item \textbf{Verify}($C, \sigma$) takes a multiset of public key and message pairs and its aggregate signature $\sigma$. Outputs 1 if the signature is valid and 0 otherwise.
\end{enumerate}

The security of the scheme is based on the difficulty of an adversary to forge a signature of a chosen message after performing $q$ signature queries to an oracle. 
More specifically, the aggregate scheme is $(t, q, \epsilon)$-secure if there is no adversary $\mathcal{A}$ capable of winning the game with probability at least $\epsilon$, performing at most $q$ queries to the oracle, and running in time at most $t$ \cite{Boneh,fault}.

\todo{
\begin{example}[Aggregation of Signatures via Bilinear Maps]\label{regular_bilinear}
As an example of an aggregate signature scheme, let us consider the one proposed by Boneh et al.~\cite{Boneh}, which is based on bilinear maps. A bilinear map is a map $e: G \times G \rightarrow G_T$, with $G, G_T$ being multiplicative cyclic groups of prime order $p$, with the bilinear property which states that for all $u, v \in G$ and positive integers $a, b$, we have that $e(u^a, v^b) = e(u,v)^{ab}$~\cite{Boneh}. The regular signature generation and verification consist of the following steps. Algorithm \textbf{KeyGen}($1^\kappa$) generates a random secret key $x \in \mathbb{Z}_p$, and the public key is computed as $v = g^x$, where $g$ is a generator of $G$. Given a message $m$, \textbf{Sign}($x, m$) computes the hash $h = h(m)$ and signature $\sigma = h^x$, with $h, \sigma \in G$. The verification algorithm receives $v, m, \sigma$, computes $h = h(m)$ and accepts $\sigma$ as valid if $e(\sigma, g) = e(h, v)$~\cite{Boneh}. This comes from the bilinear property, where $e(\sigma, g) = e(h^x, g) = e(h,g)^x = e(h,g^x) = e(h,v)$. 
	The algorithm \textbf{Agg}($C_1, C_2, \sigma_1, \sigma_2$) outputs $\sigma = \sigma_1 \times \sigma_2$, which certifies the integrity and authenticity of $C = C_1 \cup C_2 = \{(v_1,m_1), \ldots, (v_n,m_n)\}$. In general, given $n$ signatures $\sigma_1, \ldots, \sigma_n$, the aggregation consists of $\sigma = \prod_{i=1}^{n} \sigma_i$~\cite{Boneh}.
	The verification algorithm \textbf{Verify}($C, \sigma$) computes $h_i = h(m_i)$ for all messages $m_i, 1 \leq i \leq n$, and ensures that $e(\sigma, g) = \prod_{i=1}^{n} e(h_i, v_i)$. Of course, in this scheme one invalid pair $(v_i, m_i)$ would make $\sigma$ invalid. 
\end{example}		
}

{\todo{\subsection{Fault-Tolerant Aggregate Signature}\label{sec:fault}}
In order to provide fault-tolerance, the signature verification needs to output a list of valid signatures instead of just 0 or 1, so the scheme should provide \emph{list verification} instead of boolean verification.
To describe this scheme, Hartung et al.~\cite{fault} use the concepts of ``claim" $c = (pk, m)$ as a tuple of message $m$ and public key $pk$, and ``claim sequence" $C$ as a sequence of claims. A claim sequence requires an order among the claims, so each of its position $i$ may contain one claim or a placeholder $\bot$. Two claim sequences $C_1, C_2$ are defined as exclusively mergeable if for all $i$, $C_1[i] = \bot$ or $C_2[i] = \bot$\todo{. Given} \rmv{, and for a}$C_1$ and $C_2$ of length $k$ and $l$, with $k \geq l$, \todo{the merge of $C_1$ and $C_2$,} $C_1 \sqcup C_2 = (c_1, \ldots, c_k)$, is defined by 

\[c_i = \begin{cases}
 C_1[i], \text{ if }  C_2[i] = \bot, C_2[i] = C_1[i] \text{ or } i > l\\
 C_2[i], \text{ otherwise.}
\end{cases}\]

Let $C = (c_1, \ldots, c_n)$ be a claim sequence and $b \in \{0,1\}^n$ be a bit sequence that specifies a selection of indexes. Then $C[b]$ denotes a subsequence of $C$ having claim $c_j$ in position $j$ whenever $b[j] = 1$, and $\bot$ otherwise. \todo{The claim sequence is used as a way to impose an order to the claims, which is required for the process of generating a fault-tolerant aggregation signature and its verification.}
For more details, see Hartung et al.~\cite[Section 3]{fault}. The definition of \todo{signature aggregation with} list verification is given bellow.

\begin{definition}\label{listver}[Hartung et al.\cite{fault}] An aggregate signature scheme with list verification consists of four algorithms $\Sigma$:
	\begin{enumerate}
		\item \textbf{KeyGen}($1^\kappa$) takes security parameter $\kappa$ and creates a key pair $(pk, sk)$.
		\item \textbf{Sign}($sk, m$) creates a signature for message $m$ using secret key $sk$.
		\item \textbf{Agg}($C_1, C_2, \tau_1, \tau_2$) takes two exclusively mergeable claim sequences $C_1$ and $C_2$ and their corresponding signatures $\tau_1$ and $\tau_2$, and outputs an aggregate signature $\tau$ that certifies the integrity and authenticity of  the sequence $C = C_1 \sqcup C_2$.
		\item \textbf{Verify}($C, \tau$) takes a claim sequence C and its aggregate signature $\tau$ as input. Outputs the set of valid claims, which can be all the elements in C or even the empty set.
	\end{enumerate}
\end{definition}

Consider a claim sequence $C$ with $n$ claims, their corresponding signatures $\sigma_1, \ldots, \sigma_n$ with at most $d$ invalid ones, and the aggregate signature $\tau$.
The aggregate signature scheme $\Sigma$ is tolerant against $d$ errors if $\Sigma$.\textbf{Verify}($C, \tau$) outputs the set of claims that have valid signatures.
Therefore, a \emph{d-fault-tolerant} aggregate signature scheme is an aggregate signature scheme with list verification with a tolerance against $d$ errors. From now on we will use $\sigma$ to represent a standard aggregate signature, and $\tau$ for \todo{an aggregate} signature of a fault-tolerant scheme.

The security of the scheme is presented by Hartung et al. \cite{fault}.
An adversary advantage on forging an aggregate signature scheme with list verification is given by its probability of winning the following game:

\begin{itemize}
	\item \textbf{Setup:} $\mathcal{A}$ receives a random public key $pk$.
	\item \textbf{Queries:} $\mathcal{A}$ adaptively requests signatures of messages $m_i$ to an oracle using $pk$ and receives $\tau_i = Sign(sk,m_i)$.
	\item \textbf{Response:} $\mathcal{A}$ outputs a claim sequence $C^*$ and an aggregate signature $\tau^*$.
\end{itemize}

The adversary wins if there is a claim $c^* = (pk, m^*)$ in $C^*$ that belongs to the set of valid claims output by \textbf{Verify}($C^*, \tau^*$), and the signature of $m^*$ was never requested to the signature oracle. The scheme is said to be $(t, q, \epsilon)$-secure if there is no adversary that runs in time $t$, performs at most $q$ queries and wins the game above with probability at least $\epsilon$.

Hartung et al.~\cite{fault} uses $d$-CFFs to instantiate a generic fault-tolerant aggregate signature scheme. Let $\mathcal{M}$ be a $t \times n$ binary incidence matrix of a $d$-CFF, where each column is the characteristic vector of a set in the family.
 Given $\mathcal{M}$ and a set $\{\sigma_1, \ldots, \sigma_n\}$ of signatures to be aggregated, each column $j$ represents a signature $\sigma_j$, and the rows of $\mathcal{M}$ indicate which signatures will be aggregated together. We are able to identify all valid signatures as long as the amount of invalid ones does not exceed a bound $d$.

Hartung et al.~\cite[Section 4]{fault} define a fault-tolerant aggregate signature scheme based on an ordinary aggregate signature scheme $\Sigma$ that supports claims, claim sequences, and the empty signature $\lambda$ as input. We denote 
$\mathcal{M}_i$ as row $i$ of matrix $\mathcal{M}$, so $C[\mathcal{M}_i]$ represents the corresponding subsequence of a claim sequence $C$.
This scheme inherits the security of $\Sigma$, and its algorithms are presented bellow:

\begin{enumerate}\label{fault-tolerant-alg}
	\item \textbf{KeyGen}($1^\kappa$) creates a key pair $(pk, sk)$ using  \textbf{KeyGen} from $\Sigma$ and security parameter $\kappa$.
	\item \textbf{Sign}($sk, m$) receives a secret key and message, and outputs the signature given by $\Sigma.$\textbf{Sign}$(sk, m)$.
	\item \textbf{Agg}($C_1, C_2, \tau_1, \tau_2$) takes two exclusive mergeable claim sequences $C_1$ and $C_2$ and corresponding signatures $\tau_1$ and $\tau_2$, and proceeds as follows:
	\begin{enumerate}
		\item If one or both of the claim sequences $C_k$ ($k \in \{1,2\}$) contain only one claim $c$, $\tau_k$ is an individual signature $\sigma_k$. We initialize $\sigma_k$ as $\tau_k$ and expand it to a vector as follows, with $j$ equals to the index of $c$ in $C_k$:
		\[\tau_k[i] =  
		\begin{cases}
		\sigma_k, & \mbox{if } \mathcal{M}[i,j] = 1,\\
		\lambda, & \mbox{otherwise},
		\end{cases}
		\]
		for $i = 1, \ldots, t$.
		\item Once $\tau_1$ and $\tau_2$ are both vectors, we aggregate them, position by position, according to the incidence matrix $\mathcal{M}$:
		\[\tau[i] = \Sigma.\text{\textbf{Agg}}(C_1[\mathcal{M}_i], C_2[\mathcal{M}_i], \tau_1[i],  \tau_2[i]), \mbox{ for } i = 1, \ldots, t.\]
		\item Output $\tau$, which certifies the integrity and authenticity of  the sequence $C = C_1 \sqcup C_2$.
	\end{enumerate}
	\item \textbf{Verify}($C, \tau$) takes a claim sequence $C$ and the aggregate signature $\tau$, and outputs the set of valid claims. Computes $b_i = \Sigma.$\textbf{Verify}$(C[\mathcal{M}_i], \tau[i])$ 
	for each $1 \leq i \leq t$ and outputs the set of valid claims consisting of the union of each $C[\mathcal{M}_i]$ such that $b_i = 1$.
\end{enumerate} 

The following theorems are from Hartung et al.\cite[Section 4]{fault} and address the security and correctness of the scheme. For details regarding their proofs, see \cite{fault}. 
\begin{theorem}[Hartung et al.\cite{fault}]
	Let $\Sigma'$ be the aggregate signature scheme with list verification presented above. If $\Sigma'$ is based on a $d$-CFF, then it is correct and tolerant against up to $d$ errors. 
\end{theorem}	
\begin{sproof}
Considering the aggregation algorithm, we need to ensure that the verification algorithm computes the correct set of valid claims, assuming there are at most $d$ invalid ones. 
The $d$-CFF property guarantees that if a claim is valid, it will be included in an aggregation $\tau[i]$ that results in $b_i = 1$, since there will be a row $i$ that avoids all invalid claims and contains this valid claim.
 Each aggregation $\tau[i]$  that contains an invalid claim will result in $b_i = 0$.
	Therefore, $\Sigma'$.\textbf{Verify}($C, \tau$) is correct.
\end{sproof}

\begin{theorem}[Hartung et al.\cite{fault}]\label{securityProof}
If $\Sigma$ is a $(t, q, \epsilon)$-secure aggregate signature scheme, then the aggregate signature scheme with list verification above is $(t', q, \epsilon)$-secure, with $t'$ approximately the same as $t$.
\end{theorem}

\begin{example}[\todo{Fault-tolerant aggregation using CFFs}]

This example consists of $n=10$ signatures aggregated according to the $1$-CFF($5,10$) matrix $\mathcal{M}$, which allows us to identify all valid signatures as long as we have at most one invalid. For instance, if $\sigma_1$ is invalid, $\tau[1]$ and $\tau[2]$ will fail, but $\tau[3], \tau[4], \tau[5]$ prove the validity of $\sigma_i, 2\leq i \leq 10$.
\[
\mathcal{M} = \begin{pmatrix}\label{cffexample}
1 & 1 & 1 & 1& 0 & 0 & 0 & 0 & 0 &0\\
1 & 0 & 0 & 0& 1 & 1 & 1 & 0 & 0 & 0\\
0 & 1 & 0 & 0& 1 & 0 & 0 & 1 & 1 &0\\ 
0 & 0 & 1 & 0& 0 & 1 & 0 & 1 &  0  &1\\
0 & 0 & 0 & 1& 0 & 0 & 1 & 0 &  1 &1\\
\end{pmatrix}
\rightarrow
\begin{array} {l}
\tau[1] = Agg(\sigma_1, \sigma_2, \sigma_3, \sigma_4)\\
\tau[2] = Agg(\sigma_1, \sigma_5, \sigma_6, \sigma_7)\\
\tau[3] = Agg(\sigma_2, \sigma_5, \sigma_8, \sigma_9)\\
\tau[4] = Agg(\sigma_3, \sigma_6, \sigma_8, \sigma_{10})\\
\tau[5] = Agg(\sigma_4, \sigma_7, \sigma_9, \sigma_{10})\\
\end{array}
 \]
 \end{example}
 
 \todo{
 \begin{example}[Fault tolerant aggregation using CFFs based on bilinear maps]
 Consider the fault-tolerant signature aggregation that uses the CFF in Example~\ref{cffexample} and is based on aggregations using bilinear maps, as given in Example~\ref{regular_bilinear}.\\
 Let $C_1=((v_1,m_1), (v_2, m_2),\bot,\bot,\bot,\bot,\bot,\bot, (v_9,m_9),\bot)$ with associated signature $\tau_1$, and 
 $C_2(\bot, \bot,\bot,\bot,(v_5,m_5),\bot,\bot,(v_8,m_8),\bot,(v_{10},m_{10}))$ with associated signature $\tau_2$.
Since in bilinear maps the aggregated signature is the product of the corresponding signatures, using ${\cal M}$ we have:
$$\tau_1=(\sigma_1\sigma_2, \sigma_1, \sigma_2\sigma_9, \bot,\sigma_9), \ \ \  
 \tau_2=(\bot,\sigma_5,\sigma_5\sigma_8, \sigma_8\sigma_{10}, \sigma_{10}).$$
The aggregation {\bf Agg}$(C_1,C_2)$ returns $\tau=( \sigma_1\sigma_2, \sigma_1\sigma_5, \sigma_2\sigma_9\sigma_5\sigma_8, \sigma_8\sigma_{10},\sigma_9 \sigma_{10})$ which is the aggregated signature for $C=C_1\sqcup C_2$:$$C=((v_1,m_1), (v_2,m_2),\bot,\bot,(v_5,m_5),\bot,\bot,(v_8,m_8), (v_9,m_9),(v_{10},m_{10})).$$
Note that commutativity and associativity property of the cyclic groups guarantees that $\tau[3]=(\sigma_2\sigma_9)(\sigma_5\sigma_8)=
\sigma_2\sigma_5\sigma_8\sigma_9$.\\
Algorithm {\bf Verify}$(C,\tau)$ calculates $h_j=h(m_j)$ for $j\in\{1,2,5,8,9,10\}$, and  verify for every $1\leq i \leq 5$ whether  
\begin{equation}e(\tau[i],g)= \prod_{1\leq j \leq 10:{\cal M}[i,j]=1, C[j]\not=\bot} e(h_j,v_j).\label{verifyeq}\end{equation}
Assuming $\sigma_1$ is the only invalid signature, according to ${\cal M}$, Equation~(\ref{verifyeq}) will fail for $i=1,2$ since $\sigma_1$ is involved in the computations, but it will hold for $i=3,4,5$. Indeed, for example for $i=3$, using the bilinear property which implies $e(xy,z)=e(x,z)e(y,z)$ and the commutative property of cyclic groups, we get
\begin{eqnarray*}
e(\tau[3],g)&=&e( \sigma_2\sigma_9\sigma_5\sigma_8,g)
=e(h_2^{x_2}h_9^{x_9}h_5^{x_5}h_8^{x_2},g), \ \hfill {\mathrm{as\ }} \sigma_i=h_i^{x_i}\\
&=&e(h_2^{x_2},g) e(h_9^{x_9},g) e(h_5^{x_5},g) e(h_8^{x_8},g), \ \hfill {\mathrm{as}}\ e(xy,z)=e(x,z)e(y,z)\\
&=&e(h_2^{x_2},g) e(h_5^{x_5},g)e(h_8^{x_8},g), e(h_9^{x_9},g), \ \ \hfill {\mathrm{by\ commutativity\ of\ }}G_T \\
&=&e(h_2,g^{x_2}) e(h_5,g^{x_5}) e(h_8,g^{x_8})e(h_9,g^{x_9}),\ \hfill {\mathrm{bilinearity\ applied\ }} 2\times  \\
&=&e(h_2,v_2) e(h_5,v_5) e(h_8,v_8)e(h_9,v_9),  \  \hfill {\mathrm{\ as\ }}(x_j,v_j) {\mathrm{\ are\ key\ pairs}}
\end{eqnarray*}
and so we proved Equation~(\ref{verifyeq}) holds for $i=3$.
From the fact that Equation~(\ref{verifyeq}) holds for $i=3,4,5$, we can conclude $\sigma_2,\sigma_5, \sigma_8, \sigma_9$ are valid signatures for $m_2, m_5, m_8, m_9$, respectively. The CFF property guarantees that any signature not present in any of the ``valid rows'' must be invalid, so we conclude $\sigma_1$ is an invalid signature for $m_1$.

 \end{example}
  
 }

The idea of fault tolerance in signature aggregation using CFFs appeared independently 
in the master's thesis of the first author as \emph{level-d signature aggregation} \cite[Chapter 5]{masters}. A related application of CFFs to modification tolerant digital signatures can be found in \cite{thaisIndocrypt,thaisIPL}. 

\subsection{Cover-Free Family Constructions}\label{sec:CFF}
Cover-free families (CFFs) are combinatorial structures studied in the context of combinatorial group testing, and frequently used in scenarios where we need to test a set of $n$ elements to identify up to $d$ invalid ones. We use them to combine these elements into a few groups, and test the groups instead of each element.

\begin{definition}\label{CFF} A {\em set system} $\mathcal{F} = (X, \mathcal{B})$ consists of a set $X = \{x_1, \ldots, x_t\}$ with $|X| = t$, and a collection $\mathcal{B} = \{B_1, \ldots, B_n\}$ with $B_i \subseteq X, 1 \leq i \leq n,$ and $|\mathcal{B}| = n$. A {\em $d$-cover-free family}, denoted $d-$CFF$(t,n)$, is a set system such that for any subset $B_{i_0} \in \mathcal{B}$ and any other $d$ subsets $B_{i_1}, \ldots, B_{i_d} \in \mathcal{B}$, we have
\begin{equation}\label{property:cff}
B_{i_0} \nsubseteq \bigcup_{j=1}^{d}B_{i_j}.
\end{equation}	
\end{definition}	

We can represent $\mathcal{F}$ as a $t \times n$ binary incidence matrix $\mathcal{M}$ by considering the characteristic vectors of subsets in $\mathcal{B}$ as columns of $\mathcal{M}$. More precisely, $\mathcal{M}_{i,j} = 1$ if $x_i \in B_j$, and $\mathcal{M}_{i,j} = 0$ otherwise.
We will interchangeably say that $\mathcal{M}$ is $d$-CFF when its corresponding set system is $d$-CFF. 
Note that if $\mathcal{M}$ is $d$-CFF, then a matrix obtained by row \todo{permutations} and column permutations is also $d$-CFF. 

An equivalent definition of $d$-CFF is based on the existence of permutation submatrices of dimension $d + 1$ \cite{Macula}. A permutation matrix is an $n \times n$ binary matrix with exactly one ``1" per row and per column, or in other words, it is obtained by permuting the rows of the identity matrix. 

\begin{proposition}\label{id}
	A matrix $\mathcal{M}$ is $d$-CFF if and only if any set of $d+1$ columns contains a permutation sub-matrix of dimension $d+1$.
\end{proposition}
	\begin{proof}
	Take $d+1$ columns $\{{j_0}, {j_1}, \ldots, {j_d}\}$ of $\mathcal{M}$. Property (\ref{property:cff}) is true for $j_l$ in place of $i_0$ and $J = \{{j_0}, {j_1}, \ldots, {j_d}\} \setminus j_l$ in place of $\{{i_1}, \ldots, {i_d}\}$ if and only if there exists a row $i$ where $\mathcal{M}_{i, j_l} = 1$ and $\mathcal{M}_{i, j_k} = 0$ for $j_k \in J$. This is equivalent to a permutation sub-matrix of dimension $d+1$ in columns $\{{j_0}, {j_1}, \ldots, {j_d}\}$.
\end{proof}

The next propositions state relationships between sub-matrices with respect to $d$-CFF properties. Their proofs follow directly from Definition~\ref{CFF}.

\begin{proposition} Let $\mathcal{M}$ be a matrix and let $\mathcal{M}'$ be a sub-matrix of $\mathcal{M}$ formed by some of its columns. If $\mathcal{M}$ is $d$-CFF, then $\mathcal{M}'$ is also $d$-CFF.
\end{proposition}

\begin{proposition}
	Let $\mathcal{M}$ be a matrix and let $\mathcal{M}'$ be a sub-matrix of $\mathcal{M}$ formed by some of its rows.  If $\mathcal{M}'$ is $d$-CFF, then $\mathcal{M}$ is $d$-CFF.
\end{proposition}

A $1$-CFF can be constructed using Sperner's theorem.

	\begin{theorem}(Sperner~\cite{sperner})\label{sperner}
	Let $\mathcal{B}$ be a collection of subsets of $\{1, \ldots, t\}$ such that $B_1 \not\subseteq B_2$ for all $B_1, B_2 \in \mathcal{B}$. Then $|\mathcal{B}| \leq {t \choose \lfloor t/2\rfloor}$. Moreover, equality holds when $\mathcal{B}$ is the collection of all the $\lfloor t/2\rfloor$-subsets of $\{1, \ldots, t\}$.
\end{theorem}
\begin{corollary}\label{corollary:1cff}
	Given $n$ and $d=1$, let $t(n)$ be the smallest integer such that  $1$-CFF($t(n),n$) exists. Then, $t(n)=\min\{s : {s \choose \lfloor s/2 \rfloor}\geq n\}$.
\end{corollary}
\begin{proof}
	Property (\ref{property:cff}) is equivalent to the family of sets having the Sperner property. Build each column of the matrix from the characteristic vector of a distinct $\lfloor t/2\rfloor$-subset of a $t$-set, with $t=\min\{s : {s \choose \lfloor s/2 \rfloor}\geq n\}$. Theorem~\ref{sperner} guarantees $t(n)=t$.
\end{proof}	
The value $t$ grows as $\log_2 n$ as $n \rightarrow \infty$, which meets the information theoretical lower bound on the amount of bits necessary to uniquely distinguish the $n$ inputs, yielding an optimal construction.

Other constructions of CFF exist for larger $d$; see Zaverucha and Stinson~\cite[Section 3.2]{zaverucha} for a discussion on how other combinatorial objects yield good CFF methods depending on the relation of $d$ and $n$. In particular, Porat and Rothschild~\cite{PR} give a construction that yields $t = c(d+1)^2 \log n$ for a constant $c$, which for fixed $d$ is optimal in terms of meeting a lower bound $\Theta(\log n)$ (see \cite[Theorem 3]{zaverucha}). 
Next we give generalizations of two constructions by Li et al.~\cite[Theorems 3.4 and 3.5]{monotone} that allows us to build larger $d-$CFFs from smaller ones. To the best of our knowledge, these results have only been stated for $d=2$ .

\begin{definition}[Kronecker product]\label{prod}
	Let $A_k$ be an $m_k \times n_k$ binary matrix, for $k=1,2$, and  \textbf{0} be the matrix of all zeroes with same dimension as $A_2$. The product $P = A_1 \otimes A_2$ is a binary matrix such that 
\begin{multicols}{2}
	$P = \begin{pmatrix}
	P_{1,1} & \ldots P_{1,n_1}\\
	\vdots & \vdots\\
	P_{m_1,1} & \ldots P_{m_1,n_1}
	\end{pmatrix}$
	\columnbreak
	where 
	\columnbreak $P_{i,j} = \begin{cases}
	A_2, & \mbox{if } A_{1_{i,j}} = 1, \\
	\text{\textbf{0}}, & \mbox{otherwise.}
	\end{cases}
	$
\end{multicols}	
\end{definition}

 The following theorem generalizes a construction by Li et al.~\cite{monotone} given for $d=2$.

\begin{theorem}\label{dcff}
	Let $A_1$ be a $d-$CFF$(t_1, n_1)$ and $A_2$ be a $d-$CFF$(t_2, n_2)$, then $C = A_1 \otimes A_2$ is a $d-$CFF$(t_1t_2, n_1 n_2)$.
\end{theorem}
	\begin{proof}
	It is enough to show that every sub-matrix of $C$ formed by $d+1$ of its columns is $d$-CFF.
	We first establish the existence of a special sub-matrix in $C$. We define a \emph{block} of $C$ as any set of $n_2$ consecutive columns that were created by a single column of $A_1$ in the Kronecker product (Definition~\ref{prod}). Since $A_1$ is $d$-CFF, Proposition~\ref{id} implies that any $l$ blocks of $C$, $l \leq d+1$, contain a set of rows that after permutations are of the form 
	
	\[\overline{C} = \begin{pmatrix} 
	A_2 &  \ldots & \textbf{0}\\   
	& \ddots & \\ 
	\textbf{0} & \ldots & A_2 \end{pmatrix}\]
	
	In other words, $\overline{C}$ is a ``generalized permutation sub-matrix". Let\\ $\mathcal{C} = \{c_{0}, c_{1}, \ldots, c_{d}\}$ be $d+1$ columns indexes of $C$ and let $B$ be the sub-matrix of C restricted to these columns. These columns belong to $l$ blocks of $C$, where $l \leq d+1$. Therefore, we can find a sub-matrix $\overline{C}$  of $C$  of the form presented above with exactly $l$ blocks. 
	Let $D_1, D_2$ be such that $D_1 \cup D_2 = \{c_{1}, \ldots, c_{d}\}$, where $D_1$ consists of the columns that belong to the same block as $c_{0}$ and $D_2$ consists of all the other columns. Let $i$ be a row such that $B_{i, c_{0}} = 1$ and $B_{i, c_{k}} = 0$, for all $c_{k} \in D_1$; since $A_2$ is $d$-CFF, such $i$ exists. Because $B_{i, c_{0}} = 1$, we know that $B_{i, c_{l}} = 0$ for all $c_{l} \in D_2$ since $\overline{C}$ is a generalized permutation sub-matrix. Since the same is true for every column in $\mathcal{C}$, we obtain a permutation matrix of order $d+1$. By Proposition \ref{id}, $B$ is $d$-CFF. Since this is true for every choice of $B$, $C$ is $d$-CFF.
\end{proof}

Next we present a construction of $d$-CFFs  based on another result by Li et al. \cite{monotone} for $d=2$. This construction gives a better result than the one from Theorem \ref{dcff} for the cases where $s < \frac{t_1t_2-t_2}{t_1}$.

\begin{construction}\label{constdcff}
	Let $A_1$ be a $t_1 \times n_1$ binary matrix, $A_2$ be a $t_2 \times n_2$ binary matrix, and $B$ be a $s \times n_2$ binary matrix. 
	Create a matrix $P = B \otimes A_1$ as in Definition \ref{prod}. This results in $P$ with $n_2$ ``blocks" of $n_1$ columns each. For each column in block $i$, append the ith column of $A_2$, for $1 \leq i \leq n_2$. Call \emph{\textbf{Const1}($A_1, A_2, B$)} the matrix obtained.
\end{construction} 	
	\begin{theorem}\label{d-sum}
		Let $d \geq 2$, $A_1$ be a $d-$CFF$(t_1, n_1)$, $A_2$ be a $d-$CFF$(t_2, n_2)$, and $B$ be a $(d-1)-$CFF$(s, n_2)$. Then \emph{C := \textbf{Const1}($A_1, A_2, B$)} is a $d-$CFF$(st_1 + t_2, n_1 n_2)$.
	\end{theorem}
\begin{proof} 
	We need to prove that C is a $d-$CFF. Let $\mathcal{C} = \{c_{0}, c_{1}, \ldots, c_{d}\}$ be $d+1$ column indexes of $C$. Consider the $n_2$ blocks of $n_1$ consecutive columns of $C$ originating each from a single column of $B$. We divide this proof into 2 cases.\\
	\noindent
	\emph{Case 1:} The $d+1$ columns belong to $d+1$ different blocks. By the construction, each column considered contains in its lower part one distinct column from $A_2$. Since $A_2$ is $d$-CFF, Proposition 1 and 2 guarantees that $C$ is $d$-CFF.\\
	\noindent
	\emph{Case 2}: The $d+1$ columns belong to $l$ different blocks, $1 \leq l \leq d$. By considering the first part of Construction \ref{constdcff}, where we compute $B \otimes A_1$, we note that since $B$ is a $(d-1)$-CFF, Proposition \ref{id} implies we can find a generalized permutation submatrix of $C$ with a total of $l$ blocks, covering $\{c_{0}, c_{1}, \ldots, c_{d}\}$, where after permutations has the following form
	
	\[\overline{C} = \begin{pmatrix} 
	A_1 &  \ldots & \textbf{0}\\   
	& \ddots & \\ 
	\textbf{0} & \ldots & A_1 \end{pmatrix}.\] 
	
	Now let $\overline{C}_1, \overline{C}_2, \ldots, \overline{C}_l$ be the set of column indexes corresponding to each block of $\overline{C}$, and $\overline{R}_1, \overline{R}_2, \ldots, \overline{R}_l$ be sets with $t_1$ consecutive row indexes each. 
	
	Since $A_1$ is $d$-CFF, and $|\overline{C}_i \cap \mathcal{C}| \leq d+1$, Proposition \ref{id} implies that the matrix formed by column indexes $\overline{C}_i \cap \mathcal{C}$ and row indexes $\overline{R_i}$ contains a permutation submatrix which after permutations gives an identity matrix $I_i$ of dimension $|\overline{C}_i \cap \mathcal{C}|$.
	Now given the form of $\overline{C}$ we know $\overline{C}_{i_1, j_1} = 0,$ for all $i_1 \in \overline{R_i},$ for all $j_1 \in \overline{C}_j \cap \mathcal{C}$, where $1 \leq j \leq l, i \neq j$.
	
	As a consequence, there exists a submatrix of $\overline{C}$ that after permutation is of the form 
	
	\[\begin{pmatrix} 
	I_1 &  \ldots & \textbf{0}\\   
	& \ddots & \\ 
	\textbf{0} & \ldots & I_l \end{pmatrix},\] 
	\noindent	
	which is an identity matrix of dimension $d+1$. Therefore, $C$ is $d$-CFF$(st_1 + t_2, n_1 n_2)$.
\end{proof}

As a corollary, it is possible to obtain a previous result by Li et al. \cite{monotone} for the specific case of $d=2$.
\begin{corollary}\label{d2-sum}\cite[Theorem 3.5]{monotone}
Suppose there exists a $2$-CFF$(t_1, n_1)$ and a $2$-CFF$(t_2, n_2)$, then there exists a $2$-CFF$(st_1 + t_2, n_1 n_2)$ for any $s$ satisfying $s \choose {\lfloor\frac{s}{2}\rfloor}$ $\geq n_2$. 
\end{corollary}
	\begin{proof}
	Let $A_1$ be a $2$-CFF($t_1,n_1$), $A_2$ be a $2$-CFF($t_2,n_2$). Let s satisfy $s \choose {\lfloor\frac{s}{2}\rfloor}$ $\geq n_2$, and $B$ consist of  $n_2$  columns of the $1$-CFF given in the proof of Corollary \ref{corollary:1cff}. By Theorem \ref{d-sum}, Const1($A_1,A_2,B$) is a $2$-CFF$(st_1 + t_2, n_1 n_2)$.
\end{proof}

\section{Our General Unbounded Scheme Based on Nested Families}\label{sec:unbounded}
If we fix a $d-$CFFs in a fault-tolerant aggregate signature \todo{scheme}, we set a bound on the number of signatures $n$ that can be aggregated, which may not be known in advance (eg. applications in secure logging and dynamic databases). 
These unbounded applications require a sequence of $d-$CFFs that allows the increase of $n$ as necessary.
Several of these applications also deal with a large amount of signatures, and it may not be possible to save each one of them individually. So besides requiring increasing size, the $d-$CFF should also take into consideration that once aggregated, the individual signatures may not be available anymore. This raises the need of a sequence of $d-$CFFs to support all these requirements.

In order to address this problem, Hartung et al.~\cite{fault} propose the notion of a fault-tolerant unbounded scheme based on what they call a \emph{monotone} family of $d$-CFFs. It consists of using a CFF incidence matrix $\mathcal{M}^{(1)}$ until its maximum $n$ is achieved, and then jump to the next matrix.
Each new matrix $\mathcal{M}^{(l+1)}$ contains $\mathcal{M}^{(l)}$ in its upper left corner, containing a matrix of zeroes below it, in a sequence that presents a monotonicity property. 

In this section, we extend the notion of unbounded aggregation and suggest a more flexible sequence of $d-$CFFs, called \emph{nested} family. We show that the nested property is enough to be able to discard individual signatures after they are aggregated. 
This allow\todo{s} us to construct an infinite sequence of $d$-cover-free families with a better compression ratio than the ones currently known for monotone families.

In the remaining of the paper, an infinite sequence $a_1, a_2, \ldots$ is compactly denoted as $(a_l)_l$ for sets and $(a^{(l)})_l$ for matrices.

\begin{definition}\label{unbounded2}
	Let $(\mathcal{M}^{(l)})_l$ be a sequence of incidence matrices of $d$-cover-free families $(\mathcal{F}_l)_l = (X_l, \mathcal{B}_l)_l$, where the number of rows and columns of $\mathcal{M}^{(l)}$ are denoted by rows(l) and cols(l), respectively. $(\mathcal{M}^{(l)})_l$ is a \emph{nested family} of incidence matrices of $d$-CFFs, if $X_l \subseteq X_{l+1}$, rows(l) $\leq$ rows(l+1), and cols(l) $\leq$ cols(l+1), and
	\[\mathcal{M}^{(l+1)}= \begin{pmatrix} \mathcal{M}^{(l)} & Y\\ Z & W \end{pmatrix}\]
where \todo{W, Y, Z are $0$-$1$ matrices of appropriate dimensions, and} each row of $Z$ is one of the rows of $\mathcal{M}^{(l)}$, or a row of all zeros, or a row of all ones.
\end{definition}

Note that a monotone family defined in~\cite{fault} is a special case of nested family, where $Z = 0$.
The authors in~\cite{fault} use monotone families to achieve unbounded aggregation in the following way. For each $1 \leq i \leq cols(l)$, if $B_i \in \mathcal{B}_l$ and $D_i \in \mathcal{B}_{l+1}$, then $B_i = D_i$\todo{, so signature $i$ is not used in new aggregations (rows corresponding to matrix $Z=0$) and for rows correspoding to $\mathcal{M}^{(l)}$ only the old aggregations rather than individual signatures are necessary to perform aggregations with the new signatures.}
In the case of nested families, instead of $B_i = D_i$ we get $B_i \subseteq D_i$ for all $1 \leq i \leq cols(l)$. The additional property requiring that the rows of $Z$ must repeat rows of $\mathcal{M}^{(l)}$, or be trivial, is what allows us to be able to only need previous aggregations and not original signatures, \todo{even though $Z$ may not be the zero matrix}. We observe that subsequences of nested families are also nested families.

Our unbounded fault-tolerant aggregate signature scheme with nested families is defined by the following algorithms. Note that \textbf{KeyGen} and \textbf{Sign} are equal to the algorithms given in page \pageref{fault-tolerant-alg}. For  \textbf{Agg} we also create a new position $\tau[0]$ in the aggregate signature $\tau$, which holds a full aggregation of all signatures considered up to that point; \todo{this is done to permit the occurrence of rows of all ones in a future matrix $Z$, and in addition, this makes algorithm \textbf{Verify} as efficient as in traditional signatures in the case  that all signatures are valid, since in this case it aborts after one call to $\Sigma$.\textbf{Verify}$(C,\tau[0])$.}
Let  $(\mathcal{M}^{(l)})_l$ be a nested family of incidence matrices of $d$-CFFs and let $\Sigma$ be a simple aggregate signature scheme that supports claim sequences, claim placeholders, and the empty signature $\lambda$.

\noindent
\textbf{Unbounded fault tolerant aggregation with nested families}	
\begin{enumerate}\label{scheme1}
\item \textbf{KeyGen}($1^\kappa$) creates a key pair $(pk, sk)$ using  \textbf{KeyGen} from $\Sigma$ and security parameter $\kappa$.
\item \textbf{Sign}($sk, m$) receives a secret key and message, and outputs the signature given by $\Sigma.$\textbf{Sign}$(sk, m)$.
		\item \textbf{Agg}($C_1, C_2, \tau_1, \tau_2$) takes two exclusive mergeable claim sequences $C_1$ and $C_2$ and corresponding signatures $\tau_1$ and $\tau_2$, and outputs the aggregate signature $\tau$ \todo{of $C_1 \sqcup C_2$}, where $|\tau| = \text{max}\{|\tau_1|,|\tau_2|\}$.
		\begin{enumerate}
			\item Let $n_k$ be the dimension of $C_k$ for $k=1,2$, and assume w.l.o.g. that $n_1 \leq n_2$. Determine $l_k$ such that cols$(\mathcal{M}^{(l_k-1)}) < n_k \leq$ cols$(\mathcal{M}^{(l_k)})$, and denote by $t_k = \text{rows}(\mathcal{M}^{(l_k)})$, $k = 1,2$. Note that $l_1 \leq l_2$ and take $\mathcal{M}$ as the submatrix of $\mathcal{M}^{(l_2)}$ consisting of the first $n_2$ columns. Note that $\mathcal{M}=  (\begin{smallmatrix}  \mathcal{M}^{(l_1)} & Y\\ Z & W \end{smallmatrix})$, for some matrices Z, Y, W satisfying the ``nesting" properties of Definition \ref{unbounded2}. For this aggregation, $\mathcal{M}$ will be the $d$-CFF matrix that plays the same role as the fixed matrix used in the bounded scheme in page \pageref{fault-tolerant-alg}.
			\item If one or both of the claim sequences $C_k$ ($k \in \{1,2\}$) contain only one claim $c$, $\tau_k$ is an individual signature $\sigma_k$. We  expand $\tau_k$  to a vector as follows, with $j$ equals to the index of $c$ in $C_k$:
			\[\tau_k[i] =  
			\begin{cases}
			\sigma_k, & \mbox{if } i = 0 \text{ or } (\mathcal{M}[i,j] = 1 \text{ and } 1 \leq i \leq t_k),\\
			\lambda, & \mbox{otherwise}.
			\end{cases}
			\]
			\item Once $\tau_1$ and $\tau_2$ are both vectors, we aggregate them position by position according to  $\mathcal{M}$. Note that by the nested family definition we have three types of row index $i$ depending on the row type of $Z$: a row of zeros, where $\mathcal{M}[i,1]$ =\ldots= $\mathcal{M}[i,n_1] = 0$ (Type 0); a row of ones, where $\mathcal{M}[i,1]$ =\ldots= $\mathcal{M}[i,n_1] = 1$ (Type 1); and a repeated row $r$ of $\mathcal{M}^{(l_1)}$, where $\mathcal{M}[i,1] = \mathcal{M}^{(l_1)}[r,1]$ =\ldots= $\mathcal{M}[i,n_1] = \mathcal{M}^{(l_1)}[r,n_1]$ (Type 2 (r)).
			First we expand $C_1$ to $\overline{C}_1$
			having the same dimension as $C_2$, i.e. $\overline{C}_1[i] = C_1[i]$ for $1 \leq i \leq n_1$, and $\overline{C}_1[i] = \perp$ for $n_1+1 \leq i \leq n_2$, then we proceed as follows.
			
			\[\tau[0] =	\Sigma.\text{\textbf{Agg}}(\overline{C}_1, C_2, \tau_1[0], \tau_2[0])\]
			For $i = 1, \ldots \todo{,}t_1:$
			\[\tau[i] = \Sigma.\text{\textbf{Agg}}(\overline{C}_1[\mathcal{M}_i], C_2[\mathcal{M}_i], \tau_1[i],  \tau_2[i]) \]
			For $i=t_1+1, \ldots, t_2:$
			\[\tau[i] =  
			\begin{cases}
			\tau_2[i], & \text{ if $i$ is Type 0},\\
			\Sigma.\text{\textbf{Agg}}(\overline{C}_1[\mathcal{M}_i], C_2[\mathcal{M}_i], \tau_1[0],  \tau_2[i]), &  \text{ if $i$ is Type 1,}\\
			\Sigma.\text{\textbf{Agg}}(\overline{C}_1[\mathcal{M}_i], C_2[\mathcal{M}_i], \tau_1[r],  \tau_2[i]), &\text{ if $i$ is Type 2 (r).}\\
			\end{cases}
			\]
			Output $\tau$.
		\end{enumerate}
	
		\item \textbf{Verify}($C, \tau$) takes a set of public key and message pairs and the aggregate signature $\tau$ and outputs the valid claims. If $\Sigma.$\textbf{Verify}$(C, \tau[0]) = 1$, output all claims, otherwise compute $b_i = \Sigma.$\textbf{Verify}$(C[\mathcal{M}_i], \tau[i])$ 
		for each $1 \leq i \leq t_2$ and output the set of valid claims consisting of the union of each $C[\mathcal{M}_i]$ such that $b_i = 1$.
\end{enumerate}

The correctness of the aggregation and verification algorithms comes from the fact that the matrices used are $d$-CFF. For the aggregation algorithm we just need to verify that the aggregated signature computed in step (c) yields the same results as if $\mathcal{M}^{(l_2)}$ was used directly on the original signatures and apply Theorem 1.
The security of the scheme comes from Theorem \ref{securityProof}, which relies on the security of the underlying aggregate scheme $\Sigma$. 

In the next example, we illustrate step (c) in more detail.
\todo{
\begin{example}[Nested property in unbounded aggregation]
Figure~\ref{tcs:example} shows how algorithm \textbf{Agg}($C_1, C_2, \tau_1, \tau_2$) deals with each type of row in part (c). 
Assume $C_1 = (\bot, c_2, \bot, \bot, c_5), C_2 = (\bot, \bot, \bot, c_4, \bot, \bot, \bot, c_8, c_9, c_{10}))$, $n_1 = 5, n_2 = 10$, and corresponding $\tau_1$ and $\tau_2$, with $\tau_1$ a vector of $(t_1+1)$ positions created according to $\mathcal{M}^{(l_1)}$ and $\tau_2$ a vector of $(t_2+1)$ positions created according to $\mathcal{M}^{(l_2)}$. In this example, we have $\mathcal{M} = \mathcal{M}^{(l_2)}$. According to step $(c)$, we start by expanding $C_1$ to $\overline{C}_1 = (\bot, c_2, \bot, \bot, c_5,\bot, \bot, \bot, \bot, \bot)$. Then construct $\tau$ as follows. We start by computing $\tau[0]$, which holds the aggregation of all signatures considered so far:
	
	\noindent
{\footnotesize $\tau[0] = \Sigma.Agg((\bot, c_2, \bot, \bot, c_5,\bot, \bot, \bot, \bot, \bot), (\bot, \bot, \bot, c_4, \bot, \bot, \bot, c_8, c_9, c_{10})),\tau_1[0], \tau_2[0]).$}
	
	For $i = 1, \ldots, t_1$ we perform regular aggregations according to $\mathcal{M}$. For $i = t_1 + 1, \ldots, t_2$ we perform aggregations according to $\mathcal{M}$ and the type of row under $\mathcal{M}^{(l_1)}$. For instance, we compute $\tau[r],\tau[a], \tau[b], \tau[c]$ using rows $\mathcal{M}_r, \mathcal{M}_a, \mathcal{M}_b, \mathcal{M}_c$ of $\mathcal{M}$, respectively, as follows:
	
	\noindent
{\footnotesize	
\begin{eqnarray*}
\tau[r] &=&  \Sigma.Agg((\bot, c_2, \bot, \bot, \bot,\bot, \bot, \bot, \bot, \bot), (\bot, \bot, \bot, c_4, \bot, \bot, \bot, c_{8}, \bot, c_{10})), \tau_1[r], \tau_2[r] )\\
	\tau[a] &=& \tau_2[a],\\
	\tau[b] &=& \Sigma.Agg((\bot, c_2, \bot, \bot, c_5,\bot, \bot, \bot, \bot, \bot), (\bot, \bot, \bot, c_4, \bot, \bot, \bot, \bot, c_9, \bot)), \tau_1[0], \tau_2[b]),\\
	\tau[c] &=&  \Sigma.Agg((\bot, c_2, \bot, \bot, \bot,\bot, \bot, \bot, \bot, \bot), (\bot, \bot, \bot, c_4, \bot, \bot, \bot, \bot, \bot, c_{10})), \tau_1[r], \tau_2[c]).
	\end{eqnarray*}}
Aggregations of types $\tau[r]$ and $\tau[a]$ already appeared when using monotone matrices, but aggregations $\tau[b]$ and $\tau[c]$ are unique to nested families. These latter types of aggregation leverage  full aggregations stored in $\tau_1[0]$ and aggregations from previous rows (exemplified by $r$). 
\end{example}
}	

\begin{figure}[t]
\todo{
	\centering
	\includegraphics[width=6cm]{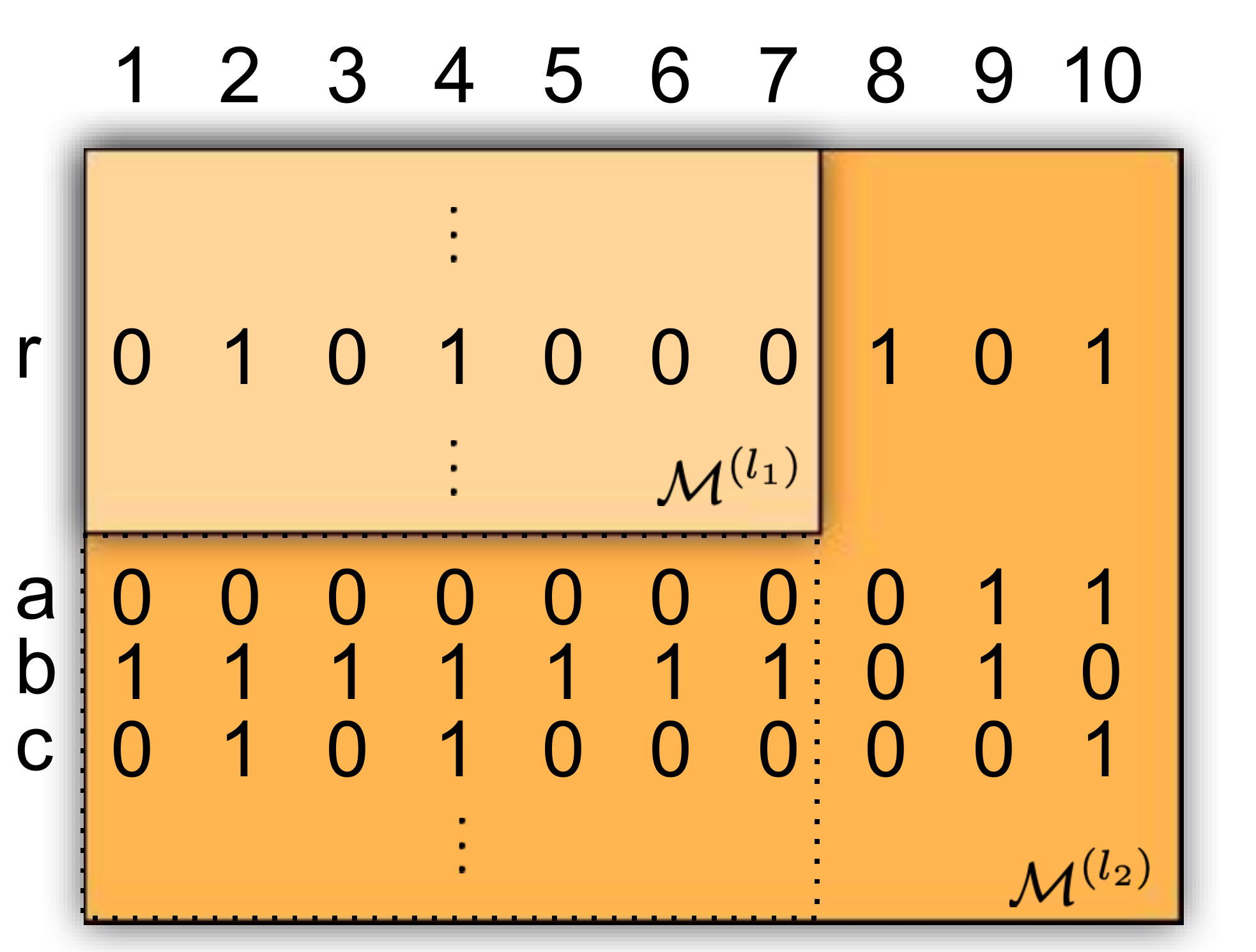}
	\caption{Example of aggregation using nested $d$-CFF matrix ${\cal M}$.}
	\label{tcs:example}
}	
\end{figure}

In the next section, we give three explicit constructions of nested families that allow us to achieve unbounded aggregation with optimal compression ratio for $d=1$ and very good compression ratios for general $d$.

\section{Construction of Nested Families with near Optimal Compression Ratios}\label{unboundedSchemes} 
Now we aim to construct a nested family of incidence matrices (Definition \ref{unbounded2}) where we can increase $n$ as necessary while avoiding to save every individual signature for further use. In this section, we propose explicit constructions of nested families for the cases where $d=1$, $d=2$, and for general values of $d$. We note that all sequences of CFF given are constructive, as they rely on ingredients that can be constructed explicitly by known methods.

\subsection{Nested Family for d=1}\label{sec:nested-d-1}

In Corollary \ref{corollary:1cff}, optimal $1$-CFFs are given based on optimal Sperner families of subsets of a $t$-set, giving a $1$-CFF($t, {t \choose \lfloor t/2 \rfloor}$) for every $t$. Since the order of elements is important for nested families, we represent the family as a tuple of sets and all we need is to order the subsets in a way to guarantee the nested properties.
Let $n_l = {l \choose \lfloor \frac{l}{2}\rfloor}$, and $\mathcal{C}_k^n$ be the list of all $k$-subsets from $\{1, \ldots, n\}$ in lexicographical order. Define $\mathcal{B}_{t}$ as follows:
\begin{align*}
\mathcal{B}_2 &= [\{1\}, \{2\}],\\
\mathcal{B}_t&=\begin{cases}
[\mathcal{B}_{t-1}[1], \ldots, \mathcal{B}_{t-1}[n_{t-1}], \mathcal{C}_{\lfloor t/2 \rfloor-1}^{t-1}[1] \cup \{t\}, \ldots, \mathcal{C}_{\lfloor t/2 \rfloor-1}^{t-1}[n_t] \cup \{t\} ],&t \mbox{ odd},\\
[\mathcal{B}_{t-1}[1] \cup \{t\}, \ldots, \mathcal{B}_{t-1}[n_{t-1}] \cup \{t\}, \mathcal{C}_{t/2}^{t-1}[1], \ldots, \mathcal{C}_{t/2}^{t-1}[n_t]  ],&t \mbox{ even},
\end{cases}
\end{align*}\label{d1const}
\noindent
for $t > 2$.

\begin{theorem}\label{theorem:d1}
	The sequence $(X_t,\mathcal{B}_t)_t$ defined above is a nested family of $1$-CFF($t, {t \choose \lfloor t/2 \rfloor}$).
\end{theorem}
\begin{proof}
	The proof that the construction gives $1$-cover-free families comes from Corollary \ref{corollary:1cff}, and the fact that each $\mathcal{B}_{t}$ has all distinct subsets of size ${t \choose \lfloor \frac{t}{2} \rfloor}$. 
	
	Next we need to show that the recursive construction of $(X_t,\mathcal{B}_t)_t$ gives a nested family. 
	For the case of $t$ odd, consider  $\mathcal{M}^{(t-1)}$ the incidence matrix of $\mathcal{F}_{t-1}$, and $A$ as the characteristic vectors of $\mathcal{C}_{\lfloor t/2 \rfloor-1}^{t-1}$. According to the recursive construction, the incidence matrix $\mathcal{M}^{(t)}$ for $\mathcal{F}_t$ has the following form: 
	
	\[\mathcal{M}^{(t)} = \begin{pmatrix} \mathcal{M}^{(t-1)} & A\\ 0 \ldots 0 & 1 \ldots 1 \end{pmatrix}\]
	
	For the case of $t$ even, with incidence matrix $\mathcal{M}^{(t-1)}$ from $\mathcal{F}_{(t-1)}$ and $A$ being the characteristic vectors of $\mathcal{C}_{t/2}^{t-1}$, the incidence matrix $\mathcal{M}^{(t)}$ of $\mathcal{F}_t$ has the following form: 
	\[\mathcal{M}^{(t)} = \begin{pmatrix} \mathcal{M}^{(t-1)} & A\\ 1 \ldots 1 & 0 \ldots 0 \end{pmatrix}\] 
	
	We note that both cases follow the form of matrix $Z$ prescribed by Definition~\ref{unbounded2} of nested family. We further observe that subsequences of $(X_t,\mathcal{B}_t)_t$ (``jumping" by more than one row) also satisfy the properties required for the corresponding $Z$ matrix in Definition~\ref{unbounded2}.
\end{proof}

It is easy to see that the property holds for larger increases of $t$ as well, where each row of $Z$ in this construction consists of all zeros or all ones. 
In other words, taking sub-sequences of $(X_t,\mathcal{B}_t)_t$ also gives a nested family.

\begin{theorem}\label{ratiod1}
	Let $(\mathcal{M}^{(l)})_l$ be the nested family defined in Theorem~\ref{theorem:d1}. This sequence has a compression ratio $\rho(n) = \frac{n}{\log_2 n}$.
\end{theorem}
\begin{proof}
	The ratio is given by $\frac{n}{t} =  \frac{{t \choose \lfloor \frac{t}{2}\rfloor}}{t}$. Using the approximation \rmv{by}\todo{of} the central binomial coefficient ${t \choose \lfloor \frac{t}{2}\rfloor} \sim \frac{2^t} {\sqrt{\pi \lfloor \frac{t}{2}\rfloor}}$, we note that $t \sim \log_2 n$. Thus, $\rho(n) = \frac{n}{\log_2 n}$.
\end{proof}
Note that this compression ratio meets the information theoretical bound, where $\log_2 n$ is a lower bound on the amount of bits that are necessary to uniquely distinguish the result of the verification algorithm.

\subsection{Nested Family for $d \geq 2$}
In this section, we give three classes of constructions of nested families.
Theorems \ref{ratiod} and \ref{ratiodsum} give specific compression ratios obtained for the cases of $d\geq2$. The one from Theorem \ref{ratiodsum} is asymptotically better, but for different constants and ranges of $n$, either one may be more
suitable. Theorem \ref{ratiod2} gives yet a better asymptotic ratio for the case $d=2$.

Theorem \ref{theorem:nested-d} gives nested families via Theorem \ref{dcff}, and Theorem \ref{ratiod} provides its corresponding ratio.
\begin{theorem}\label{theorem:nested-d} 
	Let $\mathcal{M}$ be the incidence matrix of a $d-$CFF$(t,n)$ (wlog we require $\mathcal{M}_{1,1} = 1$), and set $\mathcal{M}^{(1)} = \mathcal{M}$. We define $\mathcal{M}^{(l)} = \mathcal{M} \otimes \mathcal{M}^{(l-1)}$ for $l \geq 2$. Then $(\mathcal{M}^{(l)})_l$ is a nested family of incidence matrices of $d$-CFFs.
\end{theorem}	
\begin{proof}
	The fact that each $\mathcal{M}^{(l)}$ is $d$-CFF comes from Theorem \ref{dcff}. We need to show that $\mathcal{M}^{(l)}$ follows the structure of nested families given in Definition \ref{unbounded2}.
	First we verify that the upper-left corner of $\mathcal{M}^{(l)}$ is equal to $\mathcal{M}^{(l-1)}$\todo{;} this is true due to the requirement of $\mathcal{M}_{1,1} = 1$. Then we check that every row under $\mathcal{M}^{(l-1)}$ is either a row of $\mathcal{M}^{(l-1)}$ itself or  a row of zeros, which is true due to the Kronecker product operation. 
\end{proof}	

\begin{theorem}\label{ratiod}
	Let $d\geq 2$. Let $(\mathcal{M}^{(l)})_l$ be the nested family defined in Theorem~\ref{theorem:nested-d} using a $d$-CFF($\overline{t}, \overline{n}$) matrix $\mathcal{M}$ with $\overline{n}>\overline{t} > 1$. Then, the sequence has increasing compression ratio $\rho(n) = \frac{n}{n^{1/c}} = n^{1 - 1/c}$, for $c = \log_{\overline{t}} \overline{n} >1$ ($c$ depending on $d$).
\end{theorem}	
\begin{proof}
	Note that $\mathcal{M}^{(l)} = \mathcal{M} \times \mathcal{M}^{(l-1)}$ has $n_l = \overline{n} \times \overline{n}^{l-1} = \overline{n}^l$ columns and $t_l = (\overline{t}) \times (\overline{t})^{l-1} = (\overline{t})^l$ rows. Thus, $t_l = (\overline{t})^l = (\overline{t})^{\log_{\overline{n}} n_l} = (\overline{t})^{\log_{\overline{t}} n_l/\log_{\overline{t}} \overline{n}} =  n_l^{1/\log_{\overline{t}} \overline{n}} = n_l^{1/c}$ for $c= \log_{\overline{t}} \overline{n} > 1$. Consequently, the ratio is given by $\rho(n) = \frac{n}{n^{1/c}} = n^{1 - 1/c}$, where $c$ depends on $d$.
\end{proof}	
		
In Theorem \ref{theorem:nested-d-sum}, we show how to use the construction from Theorem~\ref{d-sum} to build a sequence of nested families for general $d$.
Then, in Theorems \ref{ratiodsum} and \ref{ratiod2} we give compression ratios by using specific $d$-CFF matrices. 

\begin{theorem}\label{theorem:nested-d-sum}
	Let $\mathcal{M}$ be a $d-$CFF$(t_1,n_1)$ matrix and set $\mathcal{M}^{(1)} = \mathcal{M}$. Let $B_i$ be a $(d-1)$-CFF$(s_i, n_i)$ matrix (wlog we require $B_{i_{1,1}} = 1$), for each $i \geq 1$. We recursively define $\mathcal{M}^{(l)} = \textbf{Const1}(\mathcal{M}^{(l-1)}, \mathcal{M}^{(l-1)}, B_{l-1})$, for $l \geq 2$, using Construction~\ref{constdcff}. The sequence of matrices $(\mathcal{M}^{(l)})_l$ is a nested family of incidence matrices of $d$-CFFs.
\end{theorem}
\begin{proof}
	The fact that $\mathcal{M}^{(l)}$ is $d-$CFF comes from Theorem \ref{d-sum}. Now we need to show that $\mathcal{M}^{(l)}$ follows the definition of nested family given in Definition \ref{unbounded2}. First, we verify that the upper left corner of $\mathcal{M}^{(l)}$ is equal to $\mathcal{M}^{(l-1)}$ due to fixing $B_{l-1_{1,1}} = 1$. Then we need to check that every partial row of $\mathcal{M}^{(l)}$ under $\mathcal{M}^{(l-1)}$ is either a row of $\mathcal{M}^{(l-1)}$, a row of zeros, or a row of ones.  We note that the first $n_{l-1}$ columns on the left side of $\mathcal{M}^{(l)}$ will be of the following form, with $c_1$ equals to the first column of $\mathcal{M}^{(l-1)}$ repeated $n_{l-1}$ times, and \textbf{0} equals to an all zero matrix:
	
	\begin{table}[h!]
		\centering
		\renewcommand{\arraystretch}{0.5}
		\begin{tabular}{|c|c|c|}
			\hline
			\multicolumn{3}{|c|}{$\mathcal{M}^{(l-1)}$} \\\hline
			\multicolumn{3}{|c|}{\tiny\vdots}                 \\ \hline
			\multicolumn{3}{|c|}{ \textbf{0}}          \\ \hline
			\multicolumn{3}{|c|}{\tiny\vdots}                \\ \hline
			\multicolumn{3}{|c|}{$\mathcal{M}^{(l-1)}$} \\ \hline
			\multicolumn{3}{|c|}{\tiny\vdots}                \\ \hline
			\small $c_1[0]$        & \tiny \ldots        &  \small $c_1[0]$        \\
			\tiny \vdots        & \tiny \ldots        & \tiny \vdots        \\
			\small $c_1[t_{l-1}]$        & \tiny \ldots        & \small $c_1[t_{l-1}]$        \\\hline
		\end{tabular}
	\end{table}
	Note that the rows of the constructed matrix are divided into two parts, in the top part we apply Kronecker product and in the bottom part we append the columns of $\mathcal{M}^{(l-1)}$. Therefore, partial rows under the $\mathcal{M}^{(l-1)}$ in the left corner in the second part are repetition of the same column $c_1$, the first column of  $\mathcal{M}^{(l-1)}$.
	In the top part, the rows are either equal to the rows of $\mathcal{M}^{(l-1)}$ or of zero matrices; in the bottom part of the construction we clearly obtain rows of all zeroes and/or all ones. This matches the requirements of matrix $Z$ in Definition \ref{unbounded2}.
\end{proof}

\begin{theorem}\label{ratiodsum}
	Let $d\geq 2$. Let $(\mathcal{M}^{(l)})_l$ be the nested family defined in Theorem~\ref{theorem:nested-d-sum} using a $d$-CFF($\overline{t}, \overline{n}$) matrix $\mathcal{M}$ with $\overline{n}>\overline{t}$. Then, there exists a $(d-1)$-CFF $B_l$ such that the sequence $(\mathcal{M}^{(l)})_l$ has ratio $\rho(n) = \frac{n}{(b\log_2 n)^{\log_2\log_2 n + D}}$ for constants  $b = \frac{2d^2 \ln \overline{n}}{\log_2 \overline{n}}$, and $D = 1-\log_2\log_2 \overline{n}$.
\end{theorem}
\begin{proof}
	Take $\mathcal{M}^{(1)} = \mathcal{M}$. Take $B_l$ as a $(d-1)$-CFF$(s_l, n_l)$ constructed using Porat and Rothschild \cite{PR} for every $l$, which gives a $s_l$ that is $\Theta(d^2 2^{l-1} \ln  \overline{n})$.
	We observe that $\mathcal{M}^{(l)}$ has $n_l = \overline{n}^{2^{l-1}}$ columns and $t_l = \overline{t} \prod_{i=1}^{l-1} (s_i + 1)$ rows. Therefore, $t_l$ is $\Theta(2^{l^2} d^{2(l-1)} (\ln \overline{n})^{l-1})$ and by using the fact that $2^{l-1} = \log_{\overline{n}} n_l$, $l-1 = \log_2 \log_{\overline{n}} n_l$, and therefore $l= \log_2 \frac{\log_2 n_l}{\log_2 \overline{n}} + 1$, we see that  $t_l$ is $\Theta((2^l d^2 \ln \overline{n})^l) =  \Theta((( \frac{\log_2 n_l}{\log_2 \overline{n}} \times 2) d^2 \ln \overline{n})^l) = \Theta((b \log_2 n_l)^{\log_2 \log_2 n_l +D})$ for constants $b = \frac{2d^2 \ln \overline{n}}{\log_2 \overline{n}}$, and $D = 1-\log_2\log_2 \overline{n}$.
	Therefore, $\rho(n) = \frac{n}{(b\log_2 n)^{\log_2\log_2 n + D}}$. 
\end{proof}	

The next theorem improves the ratio from Theorem \ref{ratiodsum} for $d=2$ by using optimal $1$-CFFs given by Corollary \ref{corollary:1cff}.

\begin{theorem}\label{ratiod2}
	Let $(\mathcal{M}^{(l)})_l$ be the nested family defined in Theorem~\ref{theorem:nested-d-sum} using a $2$-CFF($\overline{t}, \overline{n}$) matrix $\mathcal{M}$ with $\overline{n}>\overline{t}$, 
and each $B_i$ is a $1$-CFF given by Corollary \ref{corollary:1cff} for $n = \overline{n}^{2^{i-1}}$, $i \geq 1$. Then, the sequence $(\mathcal{M}^{(l)})_l$ has increasing ratio $\rho(n) = \frac{n}{(2\log_2 n)^{\log_2 \log_2 n +D}}$, for constant $D = 1-\log_2 \log_2 \overline{n}$. 
\end{theorem}	
\begin{proof}
	Take $B_i$ as a $1$-CFF$(s_i, n_i)$ constructed using Corollary \ref{corollary:1cff}, which gives a $s_i \simeq 2^{i-1}s_1 \simeq 2^{i-1} \log_2 \overline{n}$.
	We observe that $\mathcal{M}^{(l)}$ has $n_l = \overline{n}^{2^{l-1}}$ columns and $t_l = \overline{t} \prod_{i=1}^{l-1} (s_i + 1)$ rows. Then $t_l \simeq \overline{t} \prod_{i=1}^{l-1} (2^{i-1} \log_2 \overline{n}+ 1)$, which is $\Theta(2^{l^2} (\log_2 \overline{n})^{l-1})$. Now we use the fact that $2^{l-1} = \log_{\overline{n}} n_l$, $l-1 = \log_2 \log_{\overline{n}} n_l$, and therefore $l= \log_2 \frac{\log_2 n_l}{\log_2 \overline{n}} + 1$. 
	Consequently, $t_l$ is $\Theta((2^l \log_2 \overline{n})^l) = \Theta((( \frac{\log_2 n_l}{\log_2 \overline{n}} \times 2) \log_2 \overline{n})^l) = 
	\Theta((2\log_2 n_l )^{\log_2 \log_2 n_l +D})$, and the compression ratio $\rho(n) = \frac{n}{(2\log_2 n)^{\log_2 \log_2 n +D}}$ for constant $D = 1-\log_2 \log_2 \overline{n}$.
\end{proof}

\todo{
\subsection{Cost comparison of fault-tolerant and traditional signature aggregation}
In this section, we draw a brief comparison to illustrate the tradeoff between fault-tolerance and efficiency of the aggregate signature scheme. In Table~\ref{table-x}, we show a comparison of relative signature size and operation times among three signature aggregation scenarios. The value $n$ in the table refers to $n = |C|$, the number of signatures comprising the aggregation. The first scenario is traditional signature aggregation, which has zero fault-tolerance. The last scenario uses no aggregation, so that all $n$ signatures need to be transmitted and individual verification must be performed. The intermediate scenario in the middle of the table is the proposed fault-tolerant scheme that admits a constant number $d$ of invalid signatures. In this case, Table~\ref{table-x} shows, using a simple analysis of the given algorithms in page~\pageref{fault-tolerant-alg}, that signature size, aggregation time and verification time get more expensive by a factor of $(t+1)$, which depends on $n$ and $d$, as shown in Table~\ref{table-y}. We note that verification time for aggregation of valid signatures is not increased since a traditional aggregation is always kept at $\tau[0]$ and is checked first.  
}
\begin{table}[h]
\todo{
	\begin{center}
	\caption{Cost comparisons for fault-tolerant aggregation of signatures.}
	\ \\
	\begin{tabular}{l|ccc}
		             & \begin{tabular}[c]{@{}c@{}}Traditional\\ Aggregation\end{tabular}
		             & \begin{tabular}[c]{@{}c@{}}Fault-tolerant\\ aggregation\\ for fixed $d$\end{tabular} & \begin{tabular}[c]{@{}c@{}}No\\ aggregation\end{tabular} \\ \hline
		Signature size & $|\sigma|$ & $(t+1)|\sigma|$ & $n|\sigma|$           \\ \hline
		Time \textbf{KeyGen} & $x_1$     & $x_1$      & $x_1$                 \\ \hline
		Time \textbf{Sign}  & $x_2$    & $x_2$   & $x_2$  \\ \hline
		Time \textbf{Agg}  & $x_3$   & \begin{tabular}[c]{@{}c@{}}$(t+1)x_3$\end{tabular}  & N/A  \\ \hline
		Time \textbf{Verify}  & $x_4$   & $\begin{cases}
			x_4, \text{\ if no faults}\\
			\leq (t+1)x_4,\ \text{o/w}
		\end{cases}$
	& $n \cdot x_4$  \\ \hline
		Fault tolerance  & 0  & constant $d << n$  & $n$  \\ \hline
	\end{tabular}
\label{table-x}
\end{center}
}
\end{table}

\begin{table}[h]
\todo{
\begin{center}
	\caption{Values for $t$ based on $n$ and $d$.}	
	\	\\
	\begin{tabular}{l|l}
		$\mathbf{d}$& $\mathbf{t}$ \\ \hline
		$d=1$ (Theorem~\ref{ratiod1})        &$\min\{t : {t \choose \lfloor t/2 \rfloor}\geq n\} \approx \log n$    \\ \hline
		$d=2$ (Theorem~\ref{ratiod2})       &    $(2\log_2 n)^{\log_2 \log_2 n +D}$        \\ \hline
		General $d$ (Theorem~\ref{ratiodsum})  &  $(b\log_2 n)^{\log_2\log_2 n + D}$          \\ \hline
	\end{tabular}
\label{table-y}
\end{center}
}
\end{table}

\todo{
Tables~\ref{table-x} shows a clear trade-off between fault-tolerance and efficiency.  
Fault-tolerant aggregation allows for invalid signatures  without the prohibitive costs of using no aggregation, considering that in the second column the signature times and sizes grow by a factor close to $\log n$ with respect to the first column, while in third column the costs grow by a factor of $n$.

The next example considers the short signature scheme BLS~\cite{BLS}, which is based on bilinear maps and allows aggregation~\cite{Boneh}. The size of the signature comes from the BLS Standard Draft to the Internet Engineering Task Force (IETF) in~\cite{BLSIETF}.

\begin{example}
Using BLS aggregation with signatures of size $|\sigma| = 48$ bytes and fault tolerance $d=1$, the fault-tolerant aggregation of $1700$ signatures will have $14 \times 48 = 672$ bytes, while the fault-tolerant aggregation of  $1300000$ signatures will have $24 \times 48 = 1152$ bytes. Similarly, for $1700$ signatures the fault-tolerant aggregation and verification times will be multiplied by $14$ with respect to traditional aggregated signatures; while for  $1300000$ signatures this time will be multiplied by $24$. For verification time, these costs are a worst-case scenario, since when all signatures are valid then the time of fault-tolerant verification is the same as regular verification, and when some are invalid, the computation of the partial aggregations is typically faster than the single aggregation, which means $(t+1) x_4$ is an overestimation of the verification time of fault-tolerant signatures. Indeed, remember that in traditional bilinear map aggregation, the time $x_4$ involves $n$ bilinear map evaluations and $n-1$ multiplications. 
Now for the fault-tolerant aggregation, first we need to do $n$ bilinear map evaluations and then do $n-1$ multiplications to calculate $\tau[0]$ and for each row $1\leq i\leq t$ of the CFF matrix, we need to do $w_i-1$ multiplications, where $w_i$ is the number of ones in row $i$. Thus, if $w$ is the average number of ones per row, the number of multiplications used is $n-t+\sum_{i=1}^t w_i= n-t+t w=n+t(w-1)$. In the case of $d=1$, it can be verified that the nested family construction in Section~\ref{sec:nested-d-1} yields $w=n/2$, thus verification time increases by a factor of approximately $\log(n)/2$ rather than by $\log(n)$. In the given example, when one or more invalid 
aggregates are present, fault tolerant verification would be only $7$ times more costly than regular aggregation for $1700$ signatures and $12$ times more costly for $1300000$ signatures.
\end{example}
}
\section{Generalized CFFs to support corrupted aggregation tuple }\label{sec:generalized}

A $(d; \lambda)$-CFF generalizes $d$-CFFs requiring that $\lambda$ elements in a column are not covered by the union of any other $d$ columns \cite{monotone}.

	\begin{definition}\label{dlCFF}
		A set system $\mathcal{F} = (X, \mathcal{B})$ is said to be $(d; \lambda)$-cover free family ($(d; \lambda)$-CFF) if for any subset $B_{i_0} \in \mathcal{B}$ and any other $d$ subsets $B_{i_1}, \ldots, B_{i_d} \in \mathcal{B}$, we have
		\begin{equation}\label{property:dlcff}
		\bigg| B_{i_0} \Big\backslash \bigg(\bigcup_{j=1}^{d}B_{i_j}\bigg)\bigg| \geq \lambda.
		\end{equation}	
	\end{definition}	

Note that a $d$-CFF is a $(d; 1)$-CFF. 

Suppose that we use a nested family of ($d;\lambda$)-CFFs in Scheme 1.
For the verification algorithm, suppose that up to $\lambda-1$ signatures
got corrupted. Definition \ref{dlCFF} guarantees that removing $\lambda-1$ rows of the matrix gives a $(d;1)$-CFF matrix. So the verification algorithm can still identify the invalid signatures, even without $\lambda-1$ of the
aggregations.
	
We note we can generalize the $d$-CFF construction in Theorem \ref{dcff} to generate ($d; \lambda$)-CFFs, which can be used later to generate ($d; \lambda$) nested families. 
	
	\begin{theorem}\label{dlambdacff}
		Let $A_1$ be a $(d;\lambda_1)$-CFF$(t_1, n_1)$ and $A_2$ be a $(d;\lambda_2)$-CFF$(t_2, n_2)$, then $C = A_1 \otimes A_2$ is a $(d;\lambda_1\lambda_2)$-CFF$(t_1t_2, n_1 n_2)$.
	\end{theorem} 
	\begin{proof}
		Due to the fact that $A_1$ is a $(d;\lambda_1)$-CFF, we know that for any of its $d+1$ columns there will be at least $\lambda_1$ positions $i$ where $B_{i, c_{j_0}} = 1$ and $B_{i, c_{j_k}} = 0$, for all $1 \leq k \leq d$. The same is true for $A_2$ with $\lambda_2$ such positions. So, if we follow a similar proof from Theorem \ref{dcff} we will obtain a ``generalized permutation sub-matrix" of the form
		
		\[\overline{C} = \begin{pmatrix} 
		A_2 &  \ldots & \textbf{0}\\   
		\vdots &  \ldots & \vdots\\
		A_2 &  \ldots & \textbf{0}\\
		& \ddots & \\ 
		\textbf{0} & \ldots & A_2 \\
		\vdots & \ldots & \vdots \\
		\textbf{0} & \ldots & A_2 \end{pmatrix},\]
		
		where every $A_2$ appears $\lambda_1$ times in each column, and since each $A_2$ is also a $(d;\lambda_2)$-CFF, $C$ is a $(d;\lambda_1\lambda_2)$-CFF$(t_1t_2, n_1 n_2)$.
\end{proof}

We can also generalize Theorem \ref{d-sum} to generate $(d;\lambda)$-CFFs as follows.
	
	\begin{theorem}\label{d-sum-lambda} Let $A_1$ be a $(d;\lambda_1)$-CFF$(t_1, n_1)$, $A_2$ be a $(d;\lambda_2)$-CFF$(t_2, n_2)$,  $B$ be a $(d-1;\lambda_3)$-CFF$(s, n_2)$, and $\lambda = \text{min}(\lambda_1 \lambda_3, \lambda_2)$. Then \emph{C := \textbf{Const1}($A_1, A_2, B$)} is a $(d;\lambda)$-CFF$(st_1 + t_2, n_1 n_2)$.
	\end{theorem}
	\begin{proof}
		Follow a similar proof as in Theorem \ref{d-sum}. \emph{Case 1} holds because $A_2$ is a $(d;\lambda_2)$-CFF. \emph{Case 2} holds when we consider a ``generalized permutation sub-matrix" of the form
		
		\[\overline{C} = \begin{pmatrix} 
		A_1 &  \ldots & \textbf{0}\\   
		\vdots &  \ldots & \vdots\\
		A_1 &  \ldots & \textbf{0}\\
		& \ddots & \\ 
		\textbf{0} & \ldots & A_1 \\
		\vdots & \ldots & \vdots \\
		\textbf{0} & \ldots & A_1 \end{pmatrix},\]
		
		where every $A_1$ appears $\lambda_3$ times in each column, with $A_1$ a $(d;\lambda_1)$CFF.
	\end{proof}

Now we can use these two constructions to build ($d;\lambda$) nested families, which can be applied, for example,  to the problem of fault-tolerant aggregation of signatures with transmission errors that was discussed above. 

\begin{theorem}\label{theorem:nested-dlambda} 
	Let $\mathcal{M}$ be the incidence matrix of a $(d;\lambda)$-CFF$(t,n)$ (wlog we require $\mathcal{M}_{1,1} = 1$), and set $\mathcal{M}^{(1)} = \mathcal{M}$. We define $\mathcal{M}^{(l)} = \mathcal{M} \otimes \mathcal{M}^{(l-1)}$ for $l \geq 2$. Then $(\mathcal{M}^{(l)})_l$ is a nested family of incidence matrices of $(d;\lambda^l)$-CFFs, with same ratio of Theorem \ref{ratiod}.
\end{theorem}	
\begin{proof}
The proof $\mathcal{M}^{(l)}$ is $(d;\lambda^l)$-CFF comes from Theorem \ref{dlambdacff}, and the fact that in each step we are multiplying the $\lambda$s.
The proof it is a nested family follows a similar idea as the one from Theorem \ref{theorem:nested-d}. 
\end{proof}

\begin{theorem}\label{theorem:nested-dlambda-sum}
	Let $\mathcal{M}$ be a ($d;\lambda$)-CFF$(t_1,n_1)$ matrix and set $\mathcal{M}^{(1)} = \mathcal{M}$. Let $B_i$ be a $(d-1; 1)$-CFF$(s_i, n_i)$ matrix (wlog we require $B_{i_{1,1}} = 1$), for each $i \geq 1$. We recursively define $\mathcal{M}^{(l)} = \emph{\textbf{Const1}}(\mathcal{M}^{(l-1)}, \mathcal{M}^{(l-1)}, B_{l-1})$, for $l \geq 2$ and \emph{\textbf{Const1}} as defined in Construction~\ref{constdcff}. The sequence of matrices $(\mathcal{M}^{(l)})_l$ is a nested family of incidence matrices of ($d;\lambda$)-CFFs, with same ratio of Theorem \ref{ratiodsum}. The same applies for ($2;\lambda$)-CFFs with $B_i$ a $(1; 1)$-CFF from Corollary \ref{corollary:1cff}, the resulting nested family has the same ratio of Theorem \ref{ratiod2}.
\end{theorem}
\begin{proof}
The proof $\mathcal{M}^{(l)}$ is $(d;\lambda)$-CFF comes from Theorem \ref{d-sum-lambda}. The proof it is a nested family follows a similar idea as the one from Theorem \ref{theorem:nested-d-sum}.
\end{proof}

\section{Final Remarks and Open Problems}
In this work, we focus on the problem of unbounded fault-tolerant aggregation of signatures proposed by Hartung et al.~\cite{fault}. Here we define a sequence of cover-free families, called \emph{nested} families, which generalizes monotone families introduced in~\cite{fault}. We give unbounded schemes using nested families with better compression ratios than the ones currently known using monotone families.

In Section \ref{sec:unbounded}, we detail the algorithms that achieve unbounded aggregation using nested families (algorithms in page \pageref{scheme1}).
In Section \ref{unboundedSchemes}, we give explicit constructions of such families for all positive integers $d$ with compression ratios relatively close to the \rmv{information theoretical bound}\todo{best known upper bound}, and for $d=1$ the compression ratio meets this bound.

 In Section 6, we show how an extra level of fault tolerance can be achieved by using well known  ($d;\lambda$)-CFFs.  This gives the ability to handle loss or corruption of up to $\lambda -1$ aggregates of the aggregation tuple $\tau$. Our constructions and compression ratios generalize naturally to nested families of ($d;\lambda$)-CFFs.

We observe that as $n$ increases in our constructions, it would be desirable to increase the threshold value $d$ as well. However, it is not hard to show that there can not be nested families with increasing $d$. Indeed, in the definition of nested families, if ${\cal M}_{l+1}$ is $d$-CFF, then so is the first block of columns (Proposition \ref{id}); by the structure of Z we know it does not add to the cover-free Property (\ref{property:cff}), we conclude ${\cal M}_{l}$ must be $d$-CFF.  Therefore, it appears to us that there is not much hope of
increasing $d$ in fault tolerant aggregate signature schemes.

In~\cite{COCOApaper}, we investigate embedding sequences of CFF that allows us to increase $d$ and $n$, which can be useful for other cryptographical applications of CFFs.
There, we propose some constructions of embedding families based on polynomials over finite fields. Some proposed variations of these constructions give monotone families (fixed $d$) with compression ratio $\rho(n)=n^{1-1/c}$, for every integer $c\geq 2$.
We believe further studies of monotone, nested and embedding families of CFFs are important directions for modern applications of group testing  where $n$ is naturally unbounded, in particular \todo{for} applications in cryptography.\\

\noindent
\textbf{Acknowledgments.} Thais Bardini Idalino acknowledges funding by CNPq-Brazil [233697/2014-4]. Lucia Moura was supported by NSERC RGPIN-2017-05212 discovery grant.

\end{document}